\documentclass[letterpaper, 10 pt, conference]{ieeeconf}

\IEEEoverridecommandlockouts   

\usepackage{amsmath}
\usepackage{amsfonts}
\usepackage{amssymb}

\let\theoremstyle\relax
\usepackage{amsthm}
\usepackage{setspace}
\usepackage{subfigure}
\usepackage{color}
\usepackage{float}
\usepackage{graphicx}
\usepackage{lettrine}
\usepackage{pgfplots}

\usepackage{algorithm}
\usepackage{algpseudocode}

\newtheorem{theorem}{Theorem}
\newtheorem{lemma}{Lemma}

\usepackage{enumerate}

\theoremstyle{definition}
\newtheorem{ex}{Example}
\newtheorem{defin}{Definition}

\newcommand{\bmtx}{\begin{bmatrix}}
\newcommand{\emtx}{\end{bmatrix}}
\newcommand{\bsmtx}{\left[ \begin{smallmatrix}} 
\newcommand{\esmtx}{\end{smallmatrix} \right]}
\newcommand{\bmatarray}[1]{\left[\begin{array}{#1}}
\newcommand{\ematarray}{\end{array}\right]}

\newcommand{\field}[1]{\mathbb{#1}}
\newcommand{\R}{\field{R}}

\newcommand{\Csafe}{\mathcal{C}}

\begin{document}


\title{Control Barrier Functions With Unmodeled Dynamics \\
Using Integral Quadratic Constraints}


\author{Peter Seiler, Mrdjan Jankovic, and Erik Hellstrom
  \thanks{Peter Seiler is with the Electrical Engineering and Computer
    Science Department, University of Michigan, Email: {\tt\small
      pseiler@umich.edu}.  }
  \thanks{Mrdjan Jankovic and Erik Hellstrom are with Ford Research and Advanced
    Engineering, P.O. Box 2053, MD 2036 SRL, Dearborn, MI 48121, USA,
    Emails: {\tt\small \{mjankov1,jhells11\}@ford.com} }
}
\date{}
\maketitle

\begin{abstract}
  This paper presents a control design method that  achieves
  safety for systems with unmodeled dynamics at the plant input.  The
  proposed method combines control barrier functions (CBFs) and
  integral quadratic constraints (IQCs).  Simplified, low-order models
  are often used in the design of the controller.  Parasitic,
  unmodeled dynamics (e.g. actuator dynamics, time delays, etc) can
  lead to safety violations.  The proposed method bounds the
  input-output behavior of these unmodeled dynamics in the time-domain
  using an $\alpha$-IQC. The $\alpha$-IQC is then incorporated into the
  CBF constraint to ensure safety.  The approach is demonstrated
  with a simple example.
\end{abstract}


\section{Introduction}
\label{sec:intro}

This paper focuses on the design of control barrier functions (CBFs)
for systems with unmodeled dynamics at the plant input. CBFs are used
to design controllers that ensure the system remains within a safe
set~\cite{ames2017control,ames2019control}.  Simplified, low-order
models are often used for design.  However, unmodeled dynamics
(e.g. actuator lags, time delays, etc) can lead to safety violations
as shown in Section~\ref{sec:impact}.

This paper presents a method to design CBFs while accounting for
unmodeled dynamics. The approach uses the integral quadratic
constraint (IQC) framework for analysis of uncertain systems
\cite{Megretski1997}.  The main IQC result in \cite{Megretski1997}
provides frequency-domain conditions for stability of uncertain linear
time-invariant (LTI) systems.  Related results have been formulated
using time-domain dissipation inequalities \cite{Veenman13, seiler13}.
The specific IQC formulation used in this paper involves a time-domain
integral with an exponential weighting (see Section~\ref{sec:IQC}).
This is called an $\alpha$-IQC\footnote{The terminology ``$\rho$-IQC''
  was first used in \cite{Lessard2014} for the discrete-time
  case. Later the term ``$\alpha$-IQC'' was used in \cite{hu16TAC} for
  the continuous-time formulation.}  and was introduced in
\cite{Lessard2014,Boczar2015} for analysis of discrete-time
optimization algorithms.  A continuous-time formulation 
was given in \cite{hu16TAC}. Finally, $\alpha$-IQCs were
used in \cite{schwenkel21} to bound the effect of unmodeled
dynamics in the design of model predictive controllers.

There is a large literature on CBFs with a good overview in
\cite{ames2019control}.  The most closely related work on robust CBFs
is briefly summarized.  Robust control barrier functions have been
developed for guaranteeing safety in the presence of
$\mathcal{L}_\infty$ bounded disturbances
\cite{xu2015robustness,garg2021robust,jankovic2018robust,choi2021robust}
or stochastic disturbances \cite{takano2018application}.  The work in
\cite{nguyen2021robust} and \cite{buch22} considers robust CBFs to
account for variations in the model (changes to the vector fields) and
input (sector-bounded) nonlinearities, respectively.  A distinguishing
feature of our paper is the ability to handle the effect unmodeled
dynamics using $\alpha$-IQCs.  The implication is that the true state
of the plant dynamics is only partially observed, i.e. the state of
the unmodeled dynamics is not measured.  Finally, we note that
\cite{jankovic18acc} provides a method to design CBFs for systems with
known time delays.  Our proposed method can handle unknown (but
bounded) delays although with more conservatism than the approach in
\cite{jankovic18acc}.


\section{Preliminaries}


\subsection{Control Barrier Functions}


This section briefly summarizes the formulation to achieve safety using
control barrier functions
\cite{ames2017control,ames2019control}. Consider the feedback system
with plant $P$, a baseline state feedback controller $k$, and a safety
filter. The plant $P$ is assumed to be given by the
following (known) input-affine dynamics:
\begin{align}
\label{eq:Pnom}
\dot{x}(t) = f(x(t)) + g(x(t)) \, u(t), \,\,\, x(0) = x_0 
\end{align}
where $x(t) \in \R^{n_x}$ is the state,
$u(t) \in\mathcal{U} \subset \R^{n_u}$ is the control input, and
$\mathcal{U}$ defines a set of feasible control inputs.  The functions
$f:\R^{n_x}\to \R^{n_x}$, $g:\R^{n_x} \to \R^{n_x\times n_u}$, and
baseline state-feedback controller $k:\R^{n_x} \to \R^{n_u}$ are all
assumed to be locally Lipschitz continuous.

\begin{figure}[htbp]
\centering
\begin{picture}(210,55)(0,-8)
\thicklines
\put(0,20){\vector(1,0){30}}
\put(30,5){\framebox(30,30){$k(x)$}}
\put(70,25){$u_0$}
\put(60,20){\vector(1,0){30}}
\put(90,5){\framebox(40,30)}
\put(98,21){Safety}
\put(100,10){Filter}
\put(145,25){$u$}
\put(130,20){\vector(1,0){30}}
\put(20,0){\dashbox(120,40)}
\put(65,45){$k_{safe}(x)$}
\put(160,5){\framebox(30,30){$P$}}
\put(190,20){\vector(1,0){30}}
\put(200,25){$x$}
\put(205,20){\line(0,-1){35}}
\put(205,-15){\line(-1,0){205}}
\put(0,-15){\line(0,1){35}}
\end{picture}
\caption{State-feedback with safety filter}
\label{fig:SF}
\end{figure}
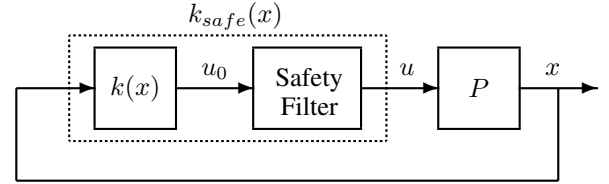
\vspace{-0.08in}


The state-feedback $k(x)$ achieves performance objectives but
is not necessarily safe. Specifically, safety is defined by a safe set
$\Csafe \subset \R^{n_x}$ and the system is in a safe state at time
$t$ if $x(t)\in \Csafe$.  We consider a safe set $\Csafe$
defined with a continuously differentiable
function $h:\R^{n_x}\to \R$:
\begin{align}
  \label{eq:Cdef}
  \Csafe:=\{ x\in \R^{n_x} \, : \, h(x) \ge 0\}
\end{align}
The boundary and interior of the safe set are denoted
$\partial \Csafe$ and $Int(\Csafe)$, respectively.  The closed-loop
dynamics with the baseline state-feedback controller are:
\begin{align}
\dot{x}(t) = f(x(t)) + g(x(t)) \, k(x(t)), \,\,\, x(0) = x_0 
\end{align}
For simplicity, assume this ordinary differential equation is forward
complete, i.e. for every initial condition there exists a unique solution
for all $t\ge 0$.  The closed-loop is said to be \emph{safe} if
$x(0)\in \Csafe$ implies $x(t)\in\Csafe$ for all $t\ge 0$. As noted
above, the closed-loop is not necessarily safe when using the baseline
state-feedback.  Control barrier functions (CBFs) are one method to
design a controller ensuring the closed loop remains in the safe set
$\Csafe$.  In particular, the function $h$ is a \emph{control barrier
function} if there exists $\alpha >0$ such that:
\begin{align}
  \label{eq:CBFconstraint}
  \sup_{u\in \mathcal{U}} \left[ L_f h(x) + L_g h(x) u \right] \ge 
  -\alpha h(x) \,\,\, \forall x\in\R^{n_x}
\end{align}
where $L_f h:= \frac{\partial h}{\partial x} f$ and
$L_g h:= \frac{\partial h}{\partial x} g$ are the Lie derivatives of
$h$ with respect to $f$ and $g$.  If $h$ is a CBF and
$x(t)\in \partial \Csafe$ then there exists $u(t)\in\mathcal{U}$
such that $\dot{h}(x(t)) \ge 0$. Thus if the state reaches
the boundary of $\Csafe$ then the control can prevent the
state from crossing out of the safe set. This is formalized in
Theorem~\ref{thm:cbf} below.  The CBF constraint
\eqref{eq:CBFconstraint} ensures that the following set of control
inputs is non-empty for all $x\in \R^{n_x}$:
\begin{align*}
  \mathcal{U}_{cbf}(x) := \{ u\in\mathcal{U} \, : \, 
  \left[ L_f h(x) + L_g h(x) u \right] \ge   -\alpha h(x) \}
\end{align*}
The existence of a control barrier function can be used to design
a controller that yields safety for the closed-loop.



\begin{theorem}\cite{ames2017control,ames2019control}
  \label{thm:cbf}
  Consider the nominal plant dynamics in \eqref{eq:Pnom}.  Let
  $\Csafe\subset \R^{n_x}$ be the superlevel set of a continuously
  differentiable function $h:\R^{n_x}\to \R$ as defined
  in~\eqref{eq:Cdef}.  Assume $h$ satisfies \eqref{eq:CBFconstraint}
  for some $\alpha>0$. Then any controller
  $k_{safe}:\R^{n_x}\to\R^{n_u}$ with
  $k_{safe}(x) \in \mathcal{U}_{cbf}(x)$ $\forall x \in \R^{n_x}$
  renders the set $\Csafe$ forward invariant.
\end{theorem}
\begin{proof}
  This is a special case of Proposition 1 and Corollary 2 in
  \cite{ames2017control}.  The trajectories $x(t)$ of the closed-loop
  with plant \eqref{eq:Pnom} and controller $k_{safe}(x)$ satisfy
  $\dot{h}(x(t)) \ge -\alpha h(x(t))$. It follows from the
  Gr\"{o}nwall-Bellman lemma~\cite{khalil01} that
  $h(t)\ge h(0) e^{-\alpha t}$ for as long as the solutions exist.
  Thus $h(0)\ge 0$ implies $h(t)\ge 0$ and the set $\Csafe$ is forward
  invariant.
\end{proof}




This summary has simplified some technical details. For example, the
CBF condition \eqref{eq:CBFconstraint} has the term $-\alpha h(x)$
where $\alpha \in \R$ is a constant.  The more general formulation in
\cite{ames2017control,ames2019control} uses $-\alpha(h(x))$ where
$\alpha$ is an extended class-$\mathcal{K}$ function. The simplifying
assumptions here allow for a proof using the Gr\"{o}nwall-Bellman
lemma. This proof will be adapted later for the case with unmodeled
dynamics. The more general results in
\cite{ames2017control,ames2019control} follow from Nagumo's theorem
\cite{blanchini15}. 


Theorem~\ref{thm:cbf} provides flexibility in the choice of the
``safe'' controller $k_{safe}$. It is useful to design a safe controller
that: (i) ensures the closed loop remains in 
$\Csafe$, and (ii) minimally alters the control command from the
baseline state-feedback. This is achieved by solving an optimization
in real-time:
\begin{align}
\label{eq:ksafe}
\begin{split}
  k_{safe}(x):= & \arg \min_{u\in \mathcal{U}} \frac{1}{2} \| u - k(x) \|^2 \\
  & \mbox{s.t. } L_f h(x) + L_g h(x) u  \ge   -\alpha h(x) 
\end{split}
\end{align}
If $\mathcal{U} = \R^{n_u}$ then \eqref{eq:ksafe} has a quadratic cost
with one linear constraint.  There is an explicit solution for this
special case.  If $\mathcal{U}$ is a polytope then \eqref{eq:ksafe} is
a quadratic program and can be efficiently solved. Finally, note that
$k_{safe}$ is not necessarily Lipschitz continuous (and the proof of
Theorem~\ref{thm:cbf} using Gr\"{o}nwall-Bellman does not
require Lipschitz continuity.)

\subsection{Impact of Unmodeled Dynamics}
\label{sec:impact}

This section presents a simple example to illustrate the impact of
unmodeled dynamics. Consider a two-dimensional point mass with
position $p\in \R^2$ and velocity $\dot{p} \in \R^2$.  A
double-integrator model for the planar motion is given by:
\begin{align}
  \dot{x}(t) = \bmtx 0 & I \\ 0 & 0 \emtx x(t) + \bmtx 0 \\ I \emtx \, u(t)
\end{align}
where $x(t)=\bsmtx p(t) \\ \dot{p}(t) \esmtx \in\R^4$ is the state
and  $u(t) \in \R^2$ contains the forces.  A baseline state-feedback
controller is designed using linear quadratic regulator with
cost matrices $Q:=diag(1,1,1.75,1.75)$ and $R:=I_2$.  This was implemented
to track a position reference command $r(t)\in \R^2$:
\begin{align*}
  u_0 = K \cdot \left(\bmtx r \\ 0 \emtx - \bmtx p \\ \dot{p} \emtx \right)
  \mbox{ where } K:=\bmtx 1 & 0 & 1.94 & 0 \\ 0 & 1 & 0 & 1.94 \emtx
\end{align*}
This baseline corresponds to independent proportional-derivative
controllers along each dimension. This differs slightly from the
feedback diagram in Figure~\ref{fig:SF} due to the inclusion of the
reference command, i.e. the baseline controller is of the form $u_0 = k(x,r)$.

A stationary obstacle of radius $\bar r=1.5$ is assumed to be at the
position $\bar c=\bsmtx 2 \\ -0.2 \esmtx$. The safe set $\Csafe$ is
defined by Equation~\ref{eq:Cdef} with
$h(x):= (p-\bar c)^T (p-\bar c) - \bar{r}^2 \ge 0$.  The time
derivatives of $h$ along a state trajectory $x$ are given by:
\begin{align}
  \label{eq:hdDI}
  \dot{h}( x(t) ) & = 2 (p(t)-\bar c)^T \dot{p}(t) \\
  \label{eq:hddDI}
  \ddot{h}( x(t), u(t) ) & = 2 (p(t)-\bar c)^T u(t) + 2 \dot{p}(t)^T \dot{p}(t) 
\end{align}
The function $h$ is not a CBF as defined in the previous section as
the control input appears in the second time derivative, i.e.  it has
relative degree 2. 

Exponential CBFs \cite{nguyen2016exponential,ames2019control} can be
used to design safe controllers for barrier functions with relative
degree greater than 1. The basic idea can be summarized as follows.
Safety is ensured if we can design a controller that achieves
$\dot{h}(t) \ge -\alpha h(t)$. Specifically,
$\dot{h}(t)\ge -\alpha h(t)$ and $h(0)\ge 0$ implies, under
appropriate technical conditions, that $h(t)\ge 0$ for as long as the
solution exists.  However, the control input $u$ does not appear in
$\dot{h}(t)$ in \eqref{eq:hdDI}.  Instead, define a new function
$\tilde{h}:= \dot{h} + \alpha h$ and note that the desired condition
is equivalent to $\tilde{h}(t) \ge 0$.  Moreover,
$\dot{\tilde{h}} = \ddot{h} + \alpha \dot{h}$. Hence the control input
appears in $\dot{\tilde{h}}$ due to \eqref{eq:hddDI}, i.e. $\tilde{h}$
is relative degree 1. Thus safety is ensured, under appropriate
technical conditions, if:
\begin{enumerate}[(i)]
\item $h(0)\ge 0$
\item $\tilde{h}(0) \ge 0$ $\Leftrightarrow$ $\dot{h}(0) \ge -\alpha h(0)$
\item $u$ is chosen so that $\dot{\tilde{h}}(t) \ge -\alpha \tilde{h}(t)$
$\Leftrightarrow$ $u$ is chosen so that 
$\ddot{h}(t) \ge  - \alpha^2 h(t) -2\alpha\dot{h}(t)$
\end{enumerate}
Roughly, conditions (ii) and (iii) ensure that $\tilde{h}(t) \ge 0$
which, combined with condition (i), ensures $h(t) \ge 0$.  The 
safe controller from the exponential CBF is obtained by solving the
following optimization in real-time:
\begin{align}
\label{eq:ksafeECBF}
\begin{split}
  k_{safe}(x,r):= & \arg \min_{u\in \mathcal{U}} \frac{1}{2} \| u - k(x,r) \|^2 \\
  & \mbox{s.t. } \ddot{h}(x,u) \ge -\alpha^2 h(x) - 2 \alpha \dot{h}(x)
\end{split}
\end{align}
Here $\dot{h}(x)$ and $\ddot{h}(x,u)$ denote the expressions in
\eqref{eq:hdDI} and \eqref{eq:hddDI}. Additional details on
exponential CBFs, including a more
rigorous derivation, can be found in
\cite{nguyen2016exponential,ames2019control}.

Figure~\ref{fig:ECBFPosition} shows a simulation of the
two-dimensional point mass with the exponential CBF controller for
$\alpha = 5$. The unsafe region due to the obstacle is shaded cyan.
The initial conditions are $p(0) = [-10,0]^T$ and
$\dot{p}(0)=[0,0]^T$.  The reference transitions linearly in time from
this initial condition to a final desired position of $[+10,0]^T$ at
time $t=45sec$.  The simulation with the nominal plant model (black
line) follows the reference and avoids the obstacle as expected.  The
figure also shows a simulation (red dashed) with the same controller
but on a plant with an input delay $\tau=0.13sec$.
The zoomed plot on the right of Figure~\ref{fig:ECBFPosition}
shows that the simulation with delay has small safety violations.
Larger delays cause even greater safety violations.

\begin{figure}[h!] 
  \centering
  \includegraphics[width=0.64\linewidth]{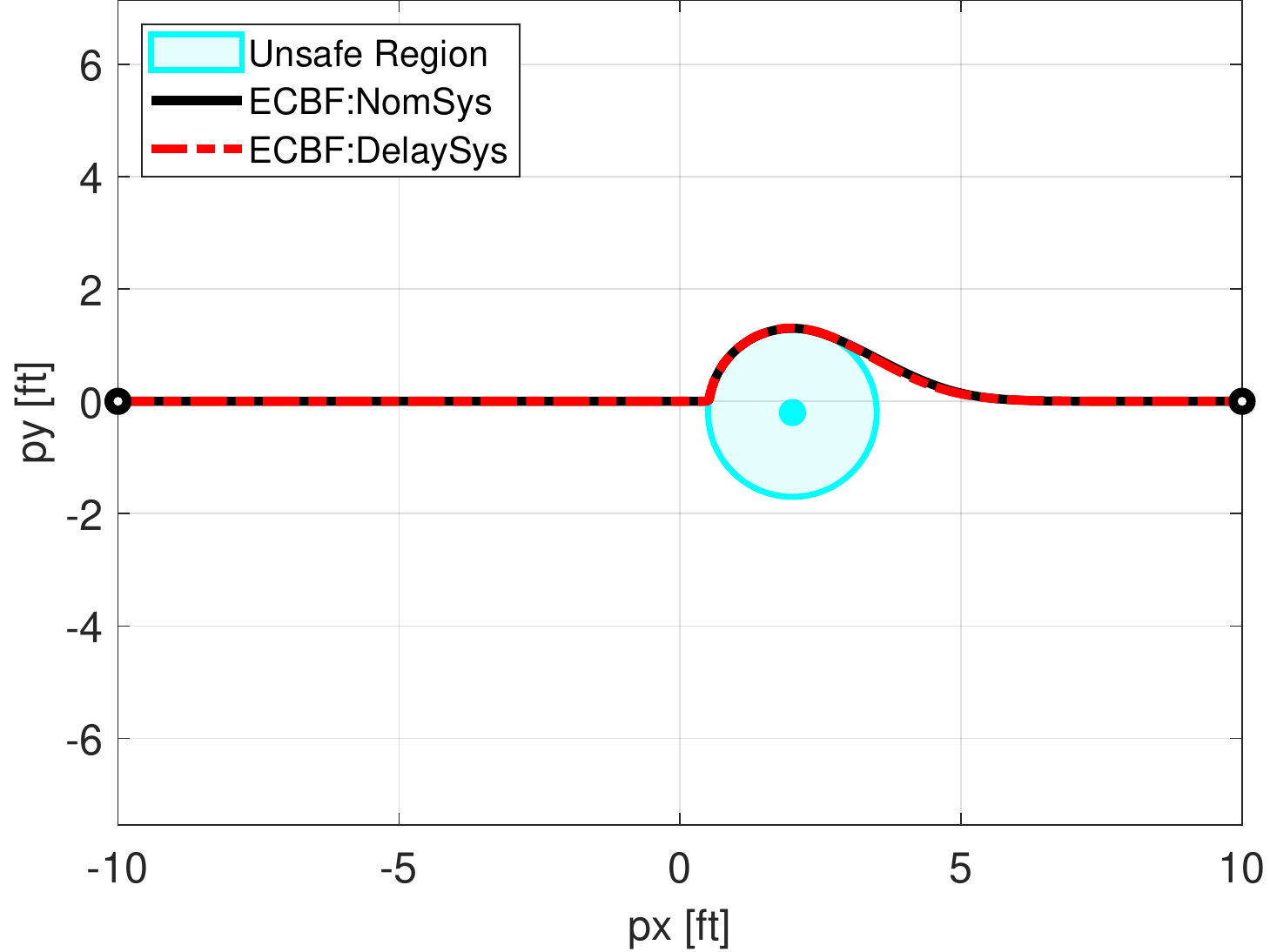}
  \includegraphics[width=0.33\linewidth,trim={1.4in, 0, 1.4in, 0},clip]{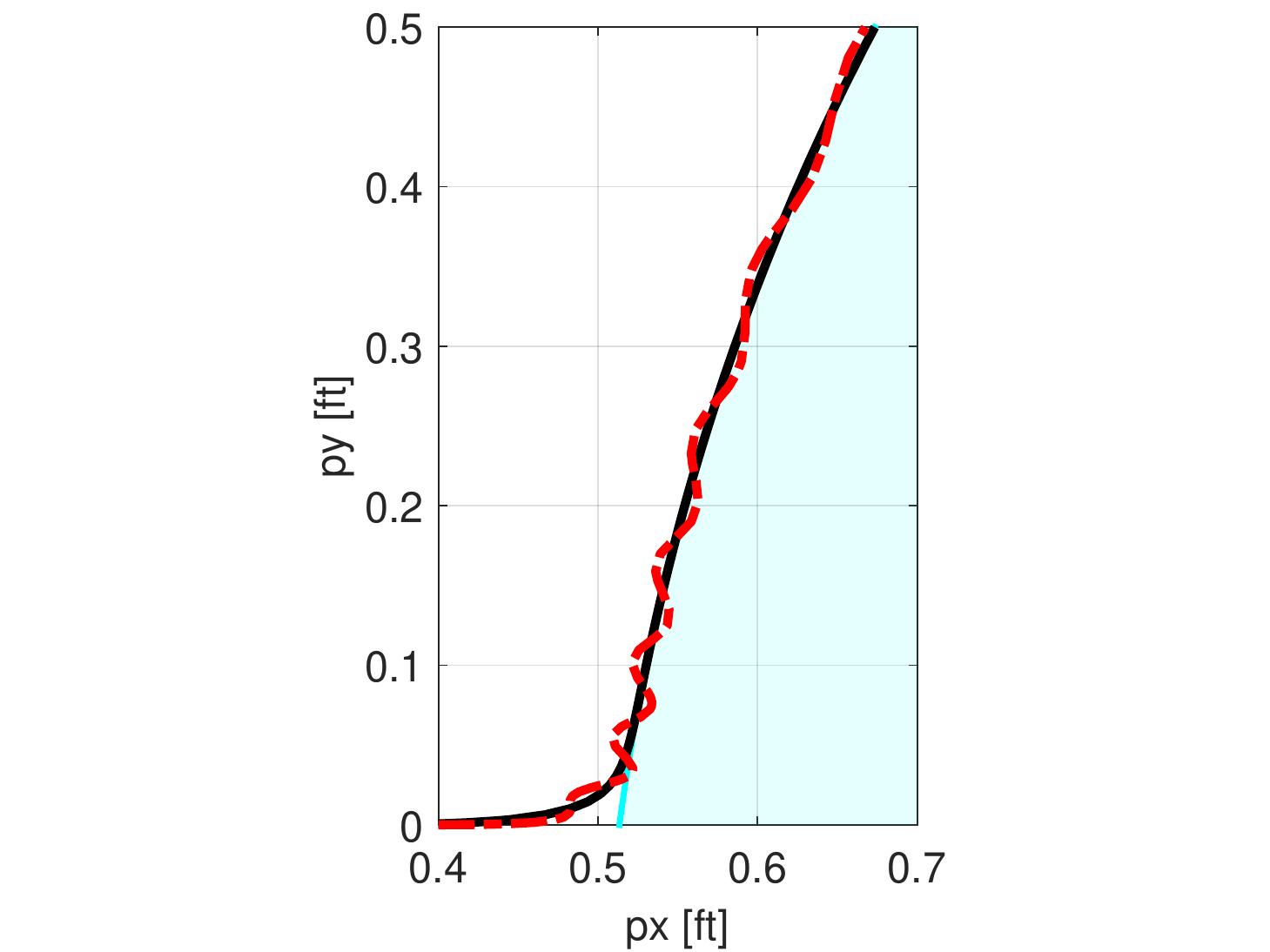}
  \caption{Position for exponential CBF controller on
    nominal point mass (black) and with additional
    delay (red dashed) of $\tau=0.13$sec.  The left plot shows full
    trajectory from $(-10,0)$ to $(10,0)$. The right plot zooms in on
    trajectories near boundary of the unsafe region.}
  \label{fig:ECBFPosition}
\end{figure}

Figure~\ref{fig:ECBFInputs} shows the the control inputs for the two
simulations.  The unmodeled delays cause the inputs to oscillate when
the exponential CBF is activated (i.e.  $k_{safe}(x,r) \ne k(x,r)$)
between $t=25sec$ to $t=32sec$.  Similar issues arise due to
unmodeled, first-order actuator dynamics.

\begin{figure}[h!] 
  \centering
  \includegraphics[width=0.9\linewidth]{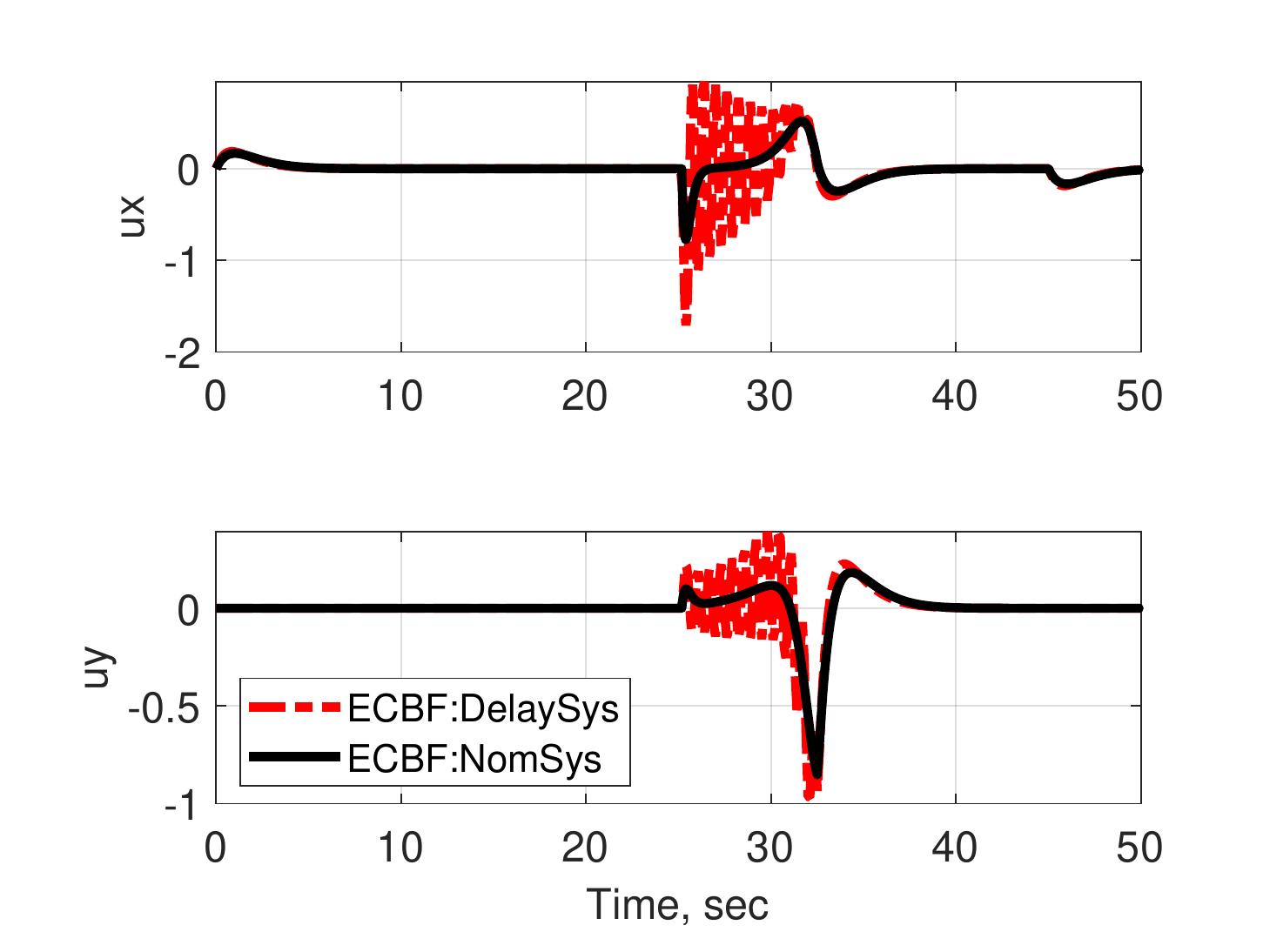}
  \caption{Control inputs for exponential CBF controller on nominal
    point mass (black) and with delay (red dashed)
    of $\tau=0.13$sec.}
  \label{fig:ECBFInputs}
\end{figure}

\section{Problem Formulation: Robust CBFs}

The safety controllers designed using CBFs or exponential CBFs are often
designed using low-order, approximate models.  This can cause issues
as indicated by the example in the previous subsection. A method to
design CBFs for systems with a known delay is given
in~\cite{jankovic18acc}.  The rest of this paper provides a method to
deal with unmodeled (unknown) delays and/or unmodeled dynamics.  In
particular, we focus on the effect of unmodeled dynamics at the plant
input.  The plant with uncertainty at the input is:
\begin{align}
\label{eq:Punc}
\begin{split}
\dot{x}(t) & = f(x(t)) + g(x(t)) \, ( u(t)+w(t) ), \,\,\, x(0) = x_0 \\
w(t) & = \Delta(u)(t)
\end{split}
\end{align}
The uncertainty enters due to the additional input $w=\Delta(u)$. If
$\Delta=0$ then this corresponds to the nominal (known) model in
Equation~\ref{eq:Pnom}.  However, $\Delta$ can have dynamics and account
for deviations from the nominal dynamics due to unmodeled
effects. This is demonstrated through two examples.

\begin{ex}[Delay]
  \label{ex:Delay}
  Assume that the actual plant input is $v=D_\tau(u)$ where $D_\tau$
  denotes a delay of $\tau$ seconds.  Thus $v=D_\tau(u)$ corresponds
  to $v(t)=u(t-\tau)$ for $t\ge \tau$ and $v(t)=0$ otherwise.  The
  effect of a delay at the plant input is modeled in
  Equation~\ref{eq:Punc} by setting $w(t):=u(t-\tau)-u(t)$.  In this
  case the perturbation is $\Delta:=D_\tau -1$.
\end{ex}

\begin{ex}[Actuator Dynamics]
  \label{ex:Actuator}
  Assume that the actual plant input is $V(s)=A(s) U(s)$ where $A(s)$
  denotes the transfer function for neglected actuator dynamics.  The
  effect of the neglected actuator dynamics at the plant input is
  modeled in \eqref{eq:Punc} by setting
  $W(s):= V(s)-U(s) = (A(s)-1)\, U(s)$.  In this case the perturbation
  is $\Delta(s):=A(s)-1$. For example, $A(s)=\frac{p}{s +p}$
  corresponds to a simple first-order model for the actuator dynamics
  yielding $\Delta(s)=\frac{-s}{s+p}$.
\end{ex}

In both examples, the signal $w$ represents the deviation from the
nominal behavior.  Note that $w$ is not simply an exogenous
disturbance as it depends on the control signal through the dynamics
of $\Delta$.  The objective is to design a safe controller that is
robust to these unmodeled dynamics.  To make this precise,
assume the unmodeled dynamics are restricted to be within a known
set $\mathbf{\Delta}$.  This set is described more formally in the
next section. For now it is sufficient to state that $\mathbf{\Delta}$
provides some bounds on the uncertainty. 


The objective is to find a condition on the control input $u$ that
ensures that the system remains safe for any uncertainty in the
uncertainty set $\mathbf{\Delta}$.  Formally, the goal is to design
$u=k_{safe}(x)$ so that the $x(0) \in \Csafe$ implies $x(t)\in\Csafe$
for all $t\ge 0$ and for all $\Delta\in\mathbf{\Delta}$.  We will use
a generalization of CBFs to ensure safety.  First note that the
nominal CBF constraint in Equation~\ref{eq:CBFconstraint} depends only
on the state $x$ and the functions $(f,g,h)$. This is an algebraic
condition that can be enforced at each time instant as part of the
optimization \eqref{eq:ksafe}.  It is important to emphasize that
the uncertainty $w=\Delta(u)$ has dynamics so that $w(t_0)$ depends,
in general, on $u(t)$ for $t\le t_0$. Our notion of robust CBF,
defined in Section~\ref{sec:RobustCBF}, will account for these dynamic
couplings.

\section{Robust CBFs with Unmodeled Dynamics}
\label{sec:RobustCBF}

\subsection{Integral Quadratic Constraints (IQCs)}
\label{sec:IQC}

Our approach relies on IQCs to bound the effect of the unmodeled
dynamics.  We use a time-domain formulation with an exponential
weighting factor. This is based on a discrete-time formulation
introduced in \cite{Lessard2014} for the analysis of optimization
algorithms. A similar formulation has also been used in
\cite{Boczar2015,hu16TAC} to analyze convergence rates and in
\cite{schwenkel21} to design robust model-predictive controllers.  A
special case of a continuous-time $\alpha$-IQC is defined
below.\footnote{Definition~\ref{def:rhoiqc} uses an exponential factor
  $e^{\alpha t}$. Continuous-time $\alpha$-IQCs have been previously
  defined using the factor $e^{2\alpha t}$ \cite{hu16TAC}. Either form
  can be converted to the other by accounting for the additional
  factor of 2.  The version used here with $e^{\alpha t}$ aligns
  closely with their use later for CBFs.}

\begin{defin}
  \label{def:rhoiqc}
  Let $F(s)$ be an $n_u \times n_u$ stable, LTI system.  A causal
  operator
  $\Delta : L_{2e}^{n_u}[0,\infty) \rightarrow L_{2e}^{n_w}[0,\infty)$
  satisfies the \underline{time-domain $\alpha$-IQC} defined by
  $F(s)$ if the following inequality holds for all
  $u \in L_{2e}^{n_u}[0,\infty)$, $w=\Delta(u)$ and $T\ge 0$
  \begin{align}
    \label{eq:rhoiqc}
    \int_0^T e^{\alpha t} \left( z(t)^T z(t) - w(t)^T w(t) \right) \, dt \ge 0 
  \end{align}
  where $z$ is the output of $F(s)$ started from zero initial
  conditions and driven by input $u$.
\end{defin}

Definition~\ref{def:rhoiqc} is a special case of a more general class
of $\alpha$-IQCs.  This special case is used for exposition and more
general $\alpha$-IQCs can be incorporated with CBFs using the method in
Section~\ref{eq:cbswithiqcs}.  The notation $\Delta \in IQC(F,\alpha)$
indicates that $\Delta$ satisfies the $\alpha$-IQC defined by $F(s)$.
The $\alpha$-IQC is a constraint on the input/output pairs of $\Delta$
and $IQC(F,\alpha)$ is the set of uncertainties bounded by the
$\alpha$-IQC. As a special case, if $\Delta$ is SISO, $\alpha=0$, and
$F(s)=1$ then \eqref{eq:rhoiqc} simplifies to
$\int_0^T w^2(t) dt \le \int_0^T u^2(t) dt$.  This represents a
constraint that the output of $\Delta$ has less energy (in the $L_2$
norm) than the input. The dynamics in $F(s)$ can be used to bound the
effect of the uncertainty as a function of frequency.  This is
demonstrated in the next example.

\begin{ex}
  \label{ex:DelayIQC}
  The uncertainty due to a delay $\tau$ is given by $w=\Delta(u)$ with
  $\Delta:=D_\tau -1$ as shown in Example~\ref{ex:Delay}. The
  $\alpha$-IQC is derived using frequency-domain relations. Let
  $U(s)$, $W(s)$, and $Z(s)$ denote the Laplace Transforms of $u(t)$,
  $w(t)$, and $z(t)$, respectively.  Thus 
  $W(s) = \Delta(s) U(s)$ and $Z(s) = F(s) U(s)$ where
  $\Delta(s)=( e^{-s\tau} -1 )$. If $\alpha=0$, we can rewrite the
  time-domain constraint \eqref{eq:rhoiqc} in the frequency domain
  using Parseval's theorem \cite{dullerud00}:
  \begin{align}
    \int_{-\infty}^\infty ( |F(j\omega)|^2 - |\Delta(j\omega)|^2) \cdot 
    |U(j\omega)|^2  \, d\omega \ge 0
  \end{align}
  This condition must hold for all inputs and hence we must select
  $F(s)$ to satisfy $|F(j\omega)| \ge |\Delta(j\omega)|$
  $\forall \omega$.  This is done by: (i) generating the frequency
  response of $\Delta(j\omega)$ for the given $\tau$, and (ii)
  computing a stable, minimum-phase $F(s)$ with
  $|F(j\omega)| \ge |\Delta(j\omega)|$ $\forall \omega$.  Step (ii)
  can be performed via convex optimization, e.g. as done in
  \texttt{fitmagfrd} in Matlab. A similar process can be used if the
  delay is unknown but restricted to $[0,\bar\tau]$ for some given
  $\bar{\tau}$. In this case, $F(s)$ is constructed to bound the
  frequency responses of $\Delta(j\omega)$ generated for many delay
  values $\tau \in [0,\bar\tau]$. This can again be solved by convex
  optimization.

  The more general case $\alpha>0$ is handled as follows. Define
  $\tilde{w}(t):=e^{\frac{\alpha}{2}t} w(t)$ and similarly for
  $\tilde{u}$ and $\tilde{z}$. Multiplication by
  $e^{\frac{\alpha}{2} t}$ in the time domain causes a shift in the
  frequency domain: $\tilde{W}(s)=W(s-\frac{\alpha}{2})$.  In
  addition, define $\tilde{\Delta}(s)=\Delta(s-\frac{\alpha}{2})$ and
  $\tilde{F}(s)=F(s-\frac{\alpha}{2})$.  Thus the shifted signals
  satisfy $\tilde{W}(s) = \tilde{\Delta}(s) \tilde{U}(s)$ and
  $\tilde{Z}(s) = \tilde{F}(s) \tilde{U}(s)$.  The shifted filter
  $\tilde{F}(s)$ can be constructed to bound the frequency response of
  $\tilde{\Delta}(s)$ as described above.  The filter
  for the $\alpha$-IQC is obtained by shifting back:
  $F(s)=\tilde{F}(s+\frac{\alpha}{2})$. These steps ensure that $F(s)$
  defines a valid $\alpha$-IQC for the delay.

\end{ex}




\subsection{CBFs with IQCs}
\label{eq:cbswithiqcs}

The effect of the uncertainty $\Delta$ can be incorporated into the
CBF condition using the $\alpha$-IQC and a Lagrange multiplier. To
elaborate on this point, assume the filter $F(s)$ has the following
state-space representation:
\begin{align}
  \begin{split}
  \dot{x}_F(t) & = A_F\, x_F(t) + B_F\, u(t),
  \,\, x_F(0)=0 \\
   z(t) & = C_F \, x_F(t) + D_F \, u(t)
 \end{split}
\end{align}
where $x_F(t) \in \R^{n_F}$ is the state of $F(s)$. The integrand in
\eqref{eq:rhoiqc} is $e^{\alpha t} I(x_F(t),u(t),w(t))$ where:
\begin{align*}
  I(x_F,u,w):= (C_F x_F + D_F u)^T (C_F x_F + D_F u ) - w^Tw
\end{align*}
The function $h$ is a \emph{robust CBF} for 
$\Delta \in IQC(F,\alpha)$ if there exists a Lagrange
multiplier $\lambda>0$ such that:
\begin{align}
  \nonumber
  & \sup_{u\in \mathcal{U}} \left[ L_f h(x) + L_g h(x) (u+w) 
    -\lambda \, I(x_F,u,w)
    \right] \ge 
    -\alpha h(x)\\
  \label{eq:RCBFconstraint}
  & \forall x\in\R^{n_x}, \, \forall x_F\in \R^{n_F},
    \, \forall w\in \R^{n_w}
\end{align}
If $h$ is a robust CBF then there exists $u(t)\in \mathcal{U}$ such
that $\dot{h}-\lambda I \ge -\alpha h$.  The following technical lemma
verifies that this is sufficient to ensure safety, i.e. $h(0)\ge 0$
implies $h(t)\ge 0$ for all time.  The lemma is stated for functions
of time and is a variation of the Gr\"{o}nwall-Bellman lemma
\cite{khalil01}.
\begin{lemma}
  \label{lem:BGversion}
  Assume $h:\R_{\geq 0} \to\R$ is continuously differentiable and
  $I:\R_{\ge 0}\to\R$ is Lebesgue integrable. In addition, assume the
  following two conditions hold for some $\alpha$, $\lambda>0$:
\begin{enumerate}
\item[(a)] $\dot{h}(t) - \lambda I(t) \ge - \alpha h(t)$ for all $t\ge 0$
\item[(b)] $\int_0^T e^{\alpha t} I(t) dt \ge 0$ for all $T \ge 0$
\end{enumerate}
Then $h(t) \ge h(0)e^{-\alpha t}$ for all $t\ge 0$.
\end{lemma}
\begin{proof}
 First, use assumption (a) to show the following:
\begin{align}
  \frac{d}{dt} \left( h(t) e^{\alpha t} \right)
  = \left( \dot{h}(t) + \alpha h(t) \right) e^{\alpha t}
  \stackrel{(a)}{\ge} \lambda e^{\alpha t} I(t)
\end{align}
Integrate this inequality from $t=0$ to $t=T$ and apply
(b):
\begin{align}
h(T) e^{\alpha T} - h(0) \ge 
\lambda \int_0^T e^{\alpha t} I(t)  \, dt \stackrel{(b)}{\ge} 0
\end{align}
This yields $h(T) \ge h(0)e^{-\alpha T}$ for all $T\ge 0$.
\end{proof}

The robust CBF constraint \eqref{eq:RCBFconstraint} ensures that the
following set is non-empty for all $x\in \R^{n_x}$ and
$x_F \in \R^{n_F}$:
\begin{align*}
 \mathcal{U}_{rcbf}(x,x_F) :=&  \{ u\in\mathcal{U} \, : \,  
L_f h(x)  + L_g h(x) (u+w)  \\
& -\lambda I(x_F,u,w)  \ge   -\alpha h(x)
  \,\, \forall w \in \R^{n_w}  \}
\end{align*}

It is emphasized that there is no a-priori bound on $w(t)$ at any
point in time.  Instead, the $\alpha$-IQC provides a bound on the
energy ($L_2$-norm) of $w$. Thus the robust CBF condition in the
definition of $\mathcal{U}_{rcbf}$ holds for all possible values of
$w$.  The next theorem states that the existence of a robust control
barrier function can be used to design a controller that yields safety
for all possible uncertainties in $IQC(F,\alpha)$.

\begin{theorem}
  \label{thm:rcbf}
  Consider the uncertain plant dynamics in \eqref{eq:Punc} with
  $\Delta \in IQC(F,\alpha)$ for some stable, LTI system $F$.  Let
  $\Csafe\subset \R^{n_x}$ be the superlevel set of a continuously
  differentiable function $h:\R^{n_x}\to \R$ as defined
  in~\eqref{eq:Cdef}.  Assume $h$ satisfies \eqref{eq:RCBFconstraint}
  for some $\alpha$, $\lambda>0$. Then any Lipschitz continuous
  controller $k_{safe}:\R^{n_x}\times \R^{n_F}\to\R^{n_u}$ with
  $k_{safe}(x,x_F) \in \mathcal{U}_{rcbf}(x,x_F)$
  $\forall (x,x_F) \in \R^{n_x}\times \R^{n_F}$ renders the set
  $\Csafe$ forward invariant for all $\Delta \in IQC(F,\alpha)$.
\end{theorem}
\begin{proof}
  The closed-loop with plant \eqref{eq:Punc}, controller
  $k_{safe}(x,x_F)$, and any $\Delta \in IQC(F,\alpha)$ has
  trajectories that satisfy the conditions (a) and (b) of
  Lemma~\ref{lem:BGversion}.  It follows from this Lemma that
  $h(0)\ge 0$ implies $h(t)\ge 0$ for as long as the solutions
  exist.
\end{proof}

The next optimization attempts to match a baseline controller $k(x)$
while satisfying the condition in Theorem~\ref{thm:rcbf}:
{\small
\begin{align*}
  & k_{safe}(x,x_F):= \arg \min_{u\in \mathcal{U}} \frac{1}{2} \| u - k(x) \|^2 
  \,\,\, \mbox{subject to:}\\
  & L_f h(x) + L_g h(x) (u+w) -\lambda I(x_F,u,w) \ge   -\alpha h(x) 
 \, \forall w \in \R^{n_w}
\end{align*}} 
The constraint is quadratic in $w$. The worst-case value of $w$ is
obtained by minimizing the left side to obtain:
\begin{align}
  \label{eq:wstar}
  w^*:= -\frac{1}{2\lambda} (L_gh(x))^T
\end{align}
The optimization can be equivalently re-written using $w^*$:
\begin{align}
\label{eq:ksafeRCBF}
  & k_{safe}(x,x_F):= \arg \min_{u\in \mathcal{U}} \frac{1}{2} \| u - k(x) \|^2 
  \,\,\, \mbox{subject to:}\\
\nonumber
  & L_f h(x) + L_g h(x) (u+w^*) -\lambda I(x_F,u,w^*) \ge   -\alpha h(x) 
\end{align}
The real-time implementation requires a measurement of the state $x$.
This can be used to form $w^*$. In addition, the filter $F(s)$ must be
simulated with input $u$ from initial condition $x_F(0)=0$ to obtain
$x_F(t)$.  Given $(x,x_F)$, the
optimization \eqref{eq:ksafeRCBF} has a convex quadratic constraint on
$u$ and a quadratic cost.  This is a convex optimization (assuming
$u\in \mathcal{U}$ is a convex constraint) and can be efficiently
solved in real-time.

Consider the special case with the following assumptions: (i) the
filter is constant with no states, i.e. $F(s)=D_F$, (ii)
$\lambda \to \infty$, and (iii) $\lambda D_F^T D_F \to 0$.  It follows
from \eqref{eq:wstar} and (ii) that $w^*$ and $-\lambda (w^*)^T w^*$
tend to zero. Moreover, (iii) implies that
$-\lambda (D_F u)^T (D_F u)$ tends to zero.  Thus the robust CBF
condition in \eqref{eq:ksafeRCBF} converges, under these assumptions,
to the nominal CBF condition in \eqref{eq:ksafe}. In other words, we
approximately recover the nominal CBF condition by choosing a small
(constant) uncertainty level for $F$ and a large value for the
Lagrange multiplier $\lambda>0$.  This provides one pragmatic approach
to handle unmodeled dynamics with CBFs: Simply use a large Lagrange
multiplier and a small constant $F$ to (heuristically) provide some
robustness to unmodeled dynamics.  A more formal approach is to bound
the unmodeled dynamics using $F(s)$ as done in
Example~\ref{ex:DelayIQC}.

Note that the optimization \eqref{eq:ksafeRCBF} is not necessarily
feasible even if $\mathcal{U} = \R^{n_u}$ due to the quadratic term
$-\lambda (C_Fx_F+D_F u)^T (C_Fx_F+D_F u)$. This is difficult to
analyze precisely as past values of $u$ impact the state $x_F$ of the
filter $F$. Smaller values of $\lambda$ tend to improve feasibility
but lead to more conservative paths around the unsafe set. Conversely,
larger values of $\lambda$ tend to degrade feasibility but more
closely approximate the performance of the nominal CBF controller.





\section{Example}
\label{sec:example}

We will again consider the two-dimensional point mass dynamics
introduced in Section~\ref{sec:impact}.  Recall that we designed an
exponential CBF with $\alpha = 5$ and explored the effect of an
unmodeled delay of $\tau = 0.13$sec.  In this section we will use a
adapt the results in Section~\ref{sec:RobustCBF} to derive a robust
exponential CBF for the two-dimensional point mass.

The first step is to derive a frequency domain bound on the
perturbation due to the unmodeled delay. We assume the true delay
$\tau$ is unknown but restricted to $[0,\bar\tau]$ with
$\bar\tau=0.13$.  The corresponding perturbation
$\Delta(s)=( e^{-s\tau} -1 )$ is bounded using the process described
in Example~\ref{ex:DelayIQC} in
Section~\ref{sec:IQC}. Figure~\ref{fig:IQCDelayBounds} shows frequency
responses (red-dashed) for
$\tilde{\Delta}(s)=\Delta(s-\frac{\alpha}{2})$ with ten values of
delay evenly spaced between $[0.1\bar\tau,\bar\tau]$.  The first-order
system $\tilde{F}(s) := \frac{2.84 s + 5.81}{s+14.48}$ satisfies
$|\tilde{F}(j\omega)| \ge |\tilde{\Delta}(j\omega)|$ $\forall \omega$
and for each delay sample. This choice of $\tilde{F}(s)$ was computed
using \texttt{fitmagfrd} in Matlab. Next, the $\alpha$-IQC filter is
obtained by shifting the frequency:
$F(s) = \tilde{F}(s+\frac{\alpha}{2})$.  The state-space data for the
resulting filter is $(A_F,B_F,C_F,D_F) = (-16.98,6.20,-5.70,2.84)$.


\begin{figure}[h!] 
  \centering
  \includegraphics[width=0.8\linewidth]{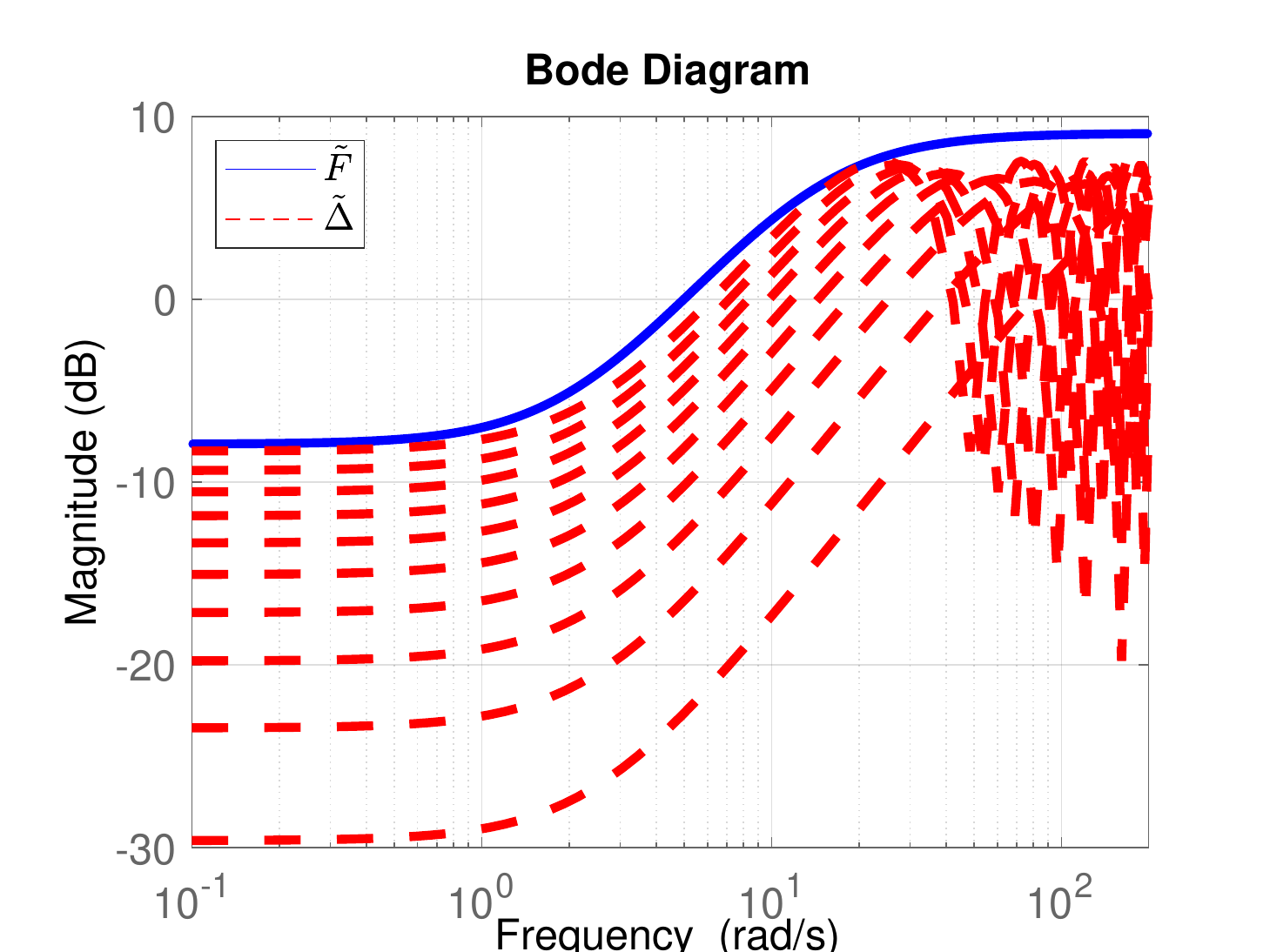}
  \caption{The delay perturbation is $\Delta(s)=( e^{-s\tau} -1 )$.
    The figure shows frequency responses of (shifted) perturbation
    $\tilde{\Delta}(s)=\Delta(s-\frac{\alpha}{2})$ with ten samples of
    delay (red dashed) and a bound $\tilde{F}$ (blue).}
  \label{fig:IQCDelayBounds}
\end{figure}

Equation~\ref{eq:ksafeECBF} gives the optimization for safe control of
the two-dimensional point mass using an exponential CBF.  This can be
adapted to include the $\alpha$-IQC using the approach in
Section~\ref{sec:RobustCBF}.  This leads to the following optimization
that merges the exponential CBF with the $\alpha$-IQC:
\begin{align}
\label{eq:ksafeECBFwithIQC}
\begin{split}
  & k_{safe}(x,x_F,r):= \arg \min_{u\in \mathcal{U}} \frac{1}{2} \| u - k(x,r) \|^2 \\
  & \mbox{s.t. } \ddot{h}(x,u) -\lambda I(x_F,u,w^*) \ge -\alpha^2 h(x) - 2 \alpha \dot{h}(x)
\end{split}
\end{align}
Here $\dot{h}(x)$ and $\ddot{h}(x,u)$ denote the expressions in
\eqref{eq:hdDI} and \eqref{eq:hddDI}.  Define
$\tilde{h} = \dot{h} + \alpha h$ so that the constraint in
\eqref{eq:ksafeECBFwithIQC} is
$\dot{\tilde{h}}-\lambda I \ge -\alpha \tilde{h}$.  It follows from
Theorem~\ref{thm:rcbf} that $\tilde{h}(0)\ge 0$ implies
$\tilde{h}(t) \ge 0$. Moreover, $\tilde{h}(t) \ge 0$ and $h(0)\ge 0$
imply $h(t)\ge 0$ based on the discussion in
Section~\ref{sec:impact}. Thus this optimization, if feasible at each
time, will yield safety.



Figure~\ref{fig:ECBFwithIQCPosition} shows the results of the nominal
exponential CBF controller (red dashed) and robust exponential CBF
(blue) on the point mass dynamics with delay of $\tau=0.13$sec.  The
plant has two inputs $(u_x,u_y)$ each of which has a delay.  An
$\alpha$-IQC for each direction was included for each delay with a
Lagrange multiplier $\lambda_x = \lambda_y =
0.1$. Figure~\ref{fig:ECBFwithIQCPosition} shows that the robust
exponential CBF controller takes a more cautious (conservative and
safe) path around the obstacle.  This accounts for the effect of the
unmodeled dynamics.

\begin{figure}[h!] 
  \centering
  \includegraphics[width=0.79\linewidth]{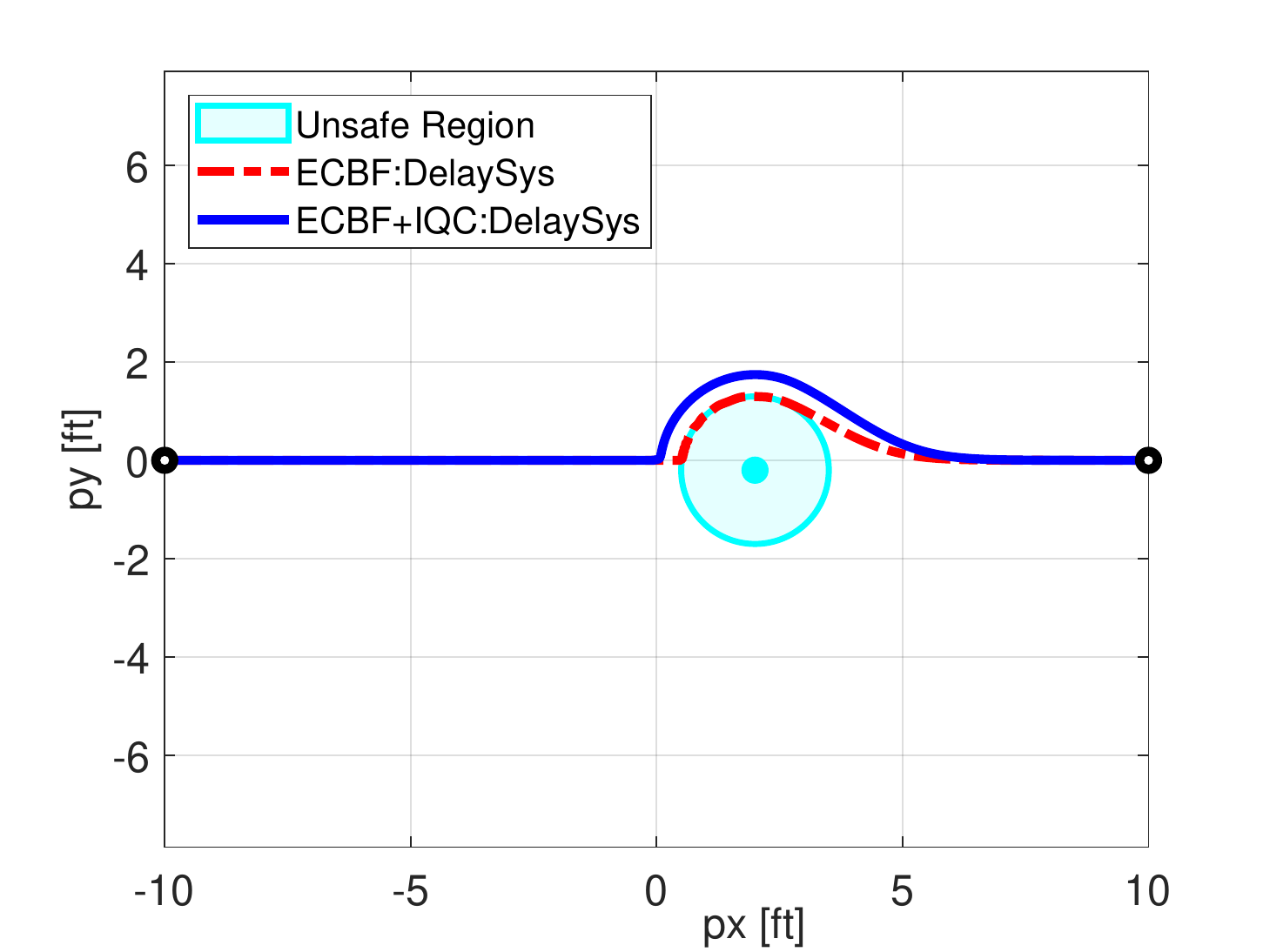}
  \caption{Position for exponential CBF (red dashed) and
    robust exponential CBF (blue) controllers on point mass
    with delay $\tau=0.13$sec.}
  \label{fig:ECBFwithIQCPosition}
\end{figure}

Figure~\ref{fig:ECBFwithIQCInputs} shows the control inputs with the
nominal exponential CBF controller (red dashed) and the robust version
(blue).  The robust version reduces the oscillations in the control
signals.  Smaller values for the Lagrange multipliers
$(\lambda_x,\lambda_y)$ further reduce the oscillations 
but also yield an even more conservative path around the
obstacle.

\begin{figure}[h!] 
  \centering
  \includegraphics[width=0.89\linewidth]{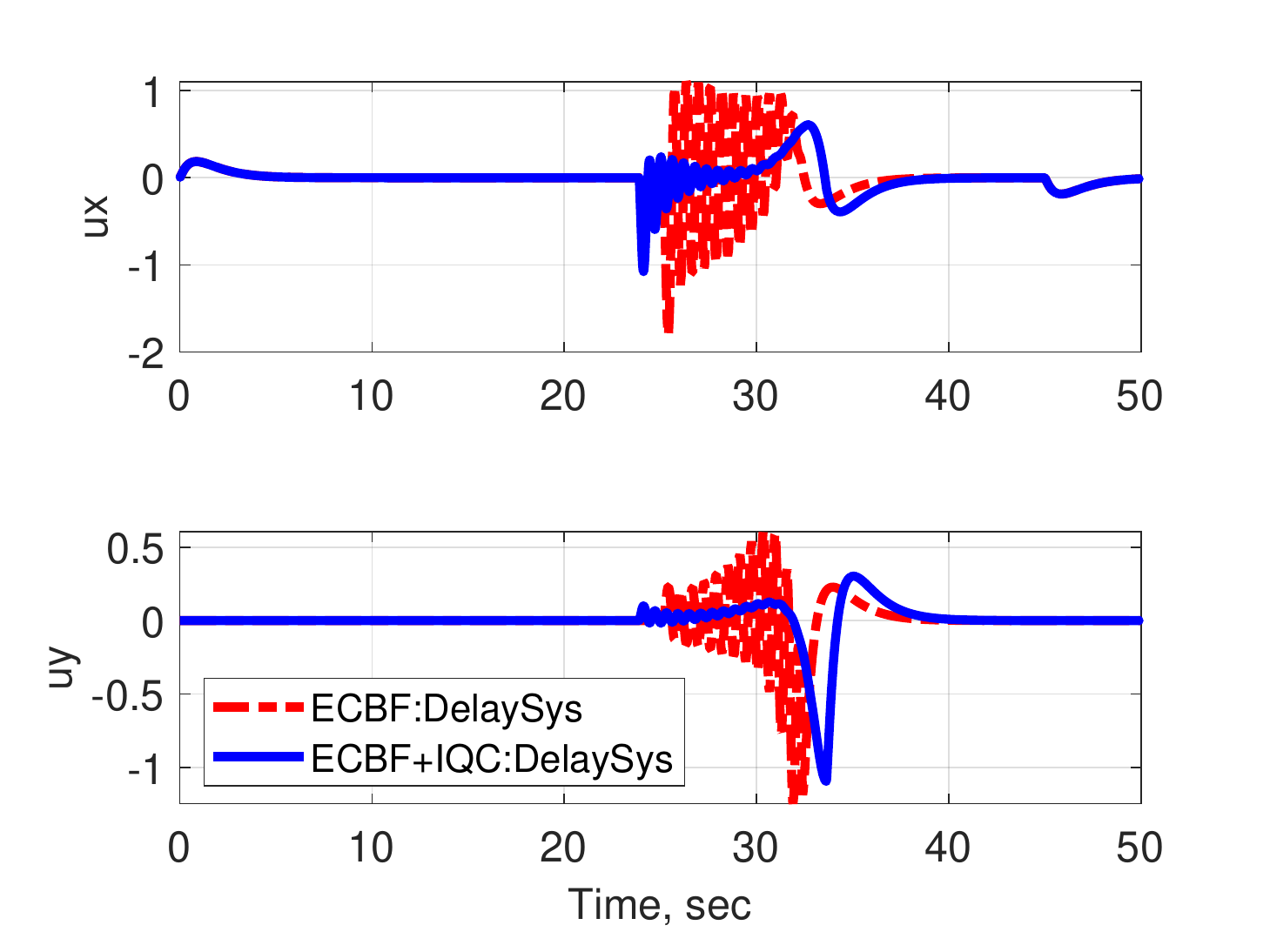}
  \caption{Control inputs for exponential CBF (red dashed) and robust
    exponential CBF (blue) controllers on point mass with
    delay $\tau=0.13$sec.}
  \label{fig:ECBFwithIQCInputs}
\end{figure}


\section{Conclusions}

This paper presented a method to design control barrier functions
(CBFs) that are robust to unmodeled dynamics at the plant input,
e.g. unmodeled actuator dynamics or time delays.  The approach uses
$\alpha$-IQCs to bound the input/output behavior of the uncertainty. A
robust CBF condition is derived using a version of the 
Gr\"{o}nwall-Bellman lemma. 

\section*{Acknowledgments}
This work was funded by the Ford/U-M Faculty Summer Sabbatical
Program. The author acknowledges useful discussions with A. Wiese,
Y. Rahman, A. Sharma, D. Sumer, M. Srinivasan, J. Buch, and
S.-C. Liao.

\bibliographystyle{IEEEtran}
\bibliography{CBFRefs}

\end{document}